\RequirePackage{amsmath}
\documentclass{llncs}
 
\usepackage{cite}
\usepackage{amssymb}
\usepackage{url}
\usepackage{graphicx}
\usepackage{siunitx}
\usepackage{threeparttable}
\usepackage{hyperref}

\usepackage[ruled, linesnumbered, vlined]{algorithm2e}
\SetKwInput{KwInput}{Input}
\SetKwInput{KwOutput}{Output}
\SetKwComment{tcc}{\# }{}
\SetKwComment{tcp}{\# }{}
\SetKwIF{If}{ElseIf}{Else}{if}{}{else if}{else}{end if}
\SetCommentSty{texttt}

\begin{document}
% TITLE - required 
\title{Sparse Tensors and Subdivision Methods for Finding the Zero Set of Polynomial Equations}
\author{Guillaume Moroz}
\institute{Universit\'e de Lorraine, CNRS, Inria,
LORIA, F-54000 Nancy, France, \email{guillaume.moroz@inria.fr}}

\maketitle

\begin{abstract}
  Finding the solutions to a system of multivariate polynomial equations is a fundamental problem in mathematics and computer science. It involves evaluating the polynomials at many points, often chosen from a grid.  In most current methods, such as subdivision, homotopy continuation, or marching cube algorithms, polynomial evaluation is treated as a black box, repeating the process for each point. We propose a new approach that partially evaluates the polynomials, allowing us to efficiently reuse computations across multiple points in a grid. Our method leverages the Compressed Sparse Fiber data structure to efficiently store and process subsets of grid points. We integrated our amortized evaluation scheme into a subdivision algorithm. Experimental results show that our approach is efficient in practice. Notably, our software \texttt{voxelize} can successfully enclose curves defined by two trivariate polynomial equations of degree $100$, a problem that was previously intractable.

\end{abstract}

% MAIN PART - Please use subsection first and without numbering
\section{Introduction}
%Interval methods are well-suited to enclose the zero set of a function $F$. 
Subdivision algorithms are widely used to enclose the zero set of a function
$F$ (\cite{JKDWbook01v0,kearfott96,Nbook90,plantinga04,Snyder1992} among others). They roughly consist in evaluating
$F$ on boxes created along a subdivision tree.
If the input function is a high degree
polynomial, one of the bottlenecks of those algorithms is
the time required to evaluate $F$.
We propose a new approach that amortizes the evaluation cost over the
boxes created in a subdivision algorithm. It combines on
the one hand partial evaluations of the input polynomial with interval
arithmetics, and on the
other hand sparse tensors \cite{SKiaaa15,CKApl18} to store the boxes created during the
subdivision algorithm.  This approach was implemented in the software
\texttt{voxelize}, and the
source code is available on
gitlab\footnote{\url{https://gitlab.inria.fr/gmoro/voxelize}}. Experimental results show that this software can enclose the zero set of polynomial
systems that were not reachable with state-of-the-art software.

After
giving an overview of our main results in the introduction, we present in
Section~\ref{sec:csf} the Compressed Sparse Fiber data structure and we
show in Section~\ref{sec:evaluationcsf} how it can be used to evaluate
efficiently a polynomial on a subset of a grid of boxes. Then in
Section~\ref{sec:fft}, we show how our new evaluation scheme yields a
quasi-linear time algorithm to compute a discrete Fourier transform.
We show in Section~\ref{sec:subdivision} how to integrate our
evaluation scheme into a subdivision algorithm to enclose the zero set
of a polynomial system. Finally, in Section~\ref{sec:experiments}, we
present the timing results of \texttt{voxelize} on several polynomial
systems, including random polynomial systems (Section~\ref{sec:random}), and
systems coming from applications (Section~\ref{sec:applications}).

\begin{figure}
  \begin{minipage}{0.4\textwidth}
  \includegraphics[width=\textwidth]{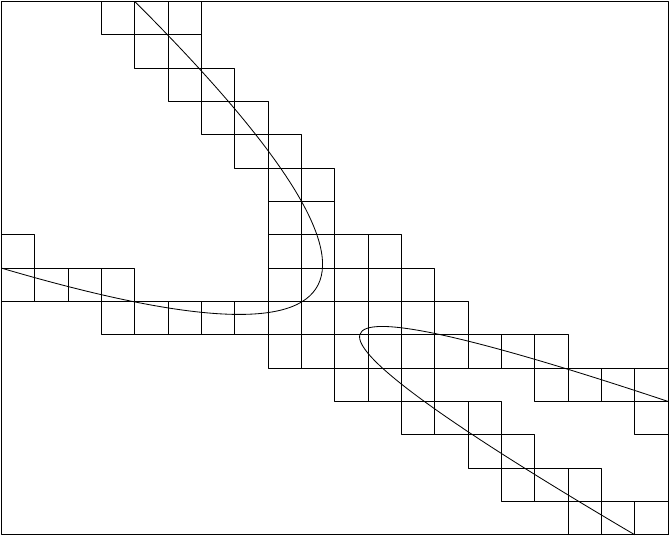}
  \caption{Boxes on the same level of the subdivision tree}
  \label{fig:2D}
\end{minipage}
\hfill
\begin{minipage}{0.4\textwidth}
  \includegraphics[width=\textwidth]{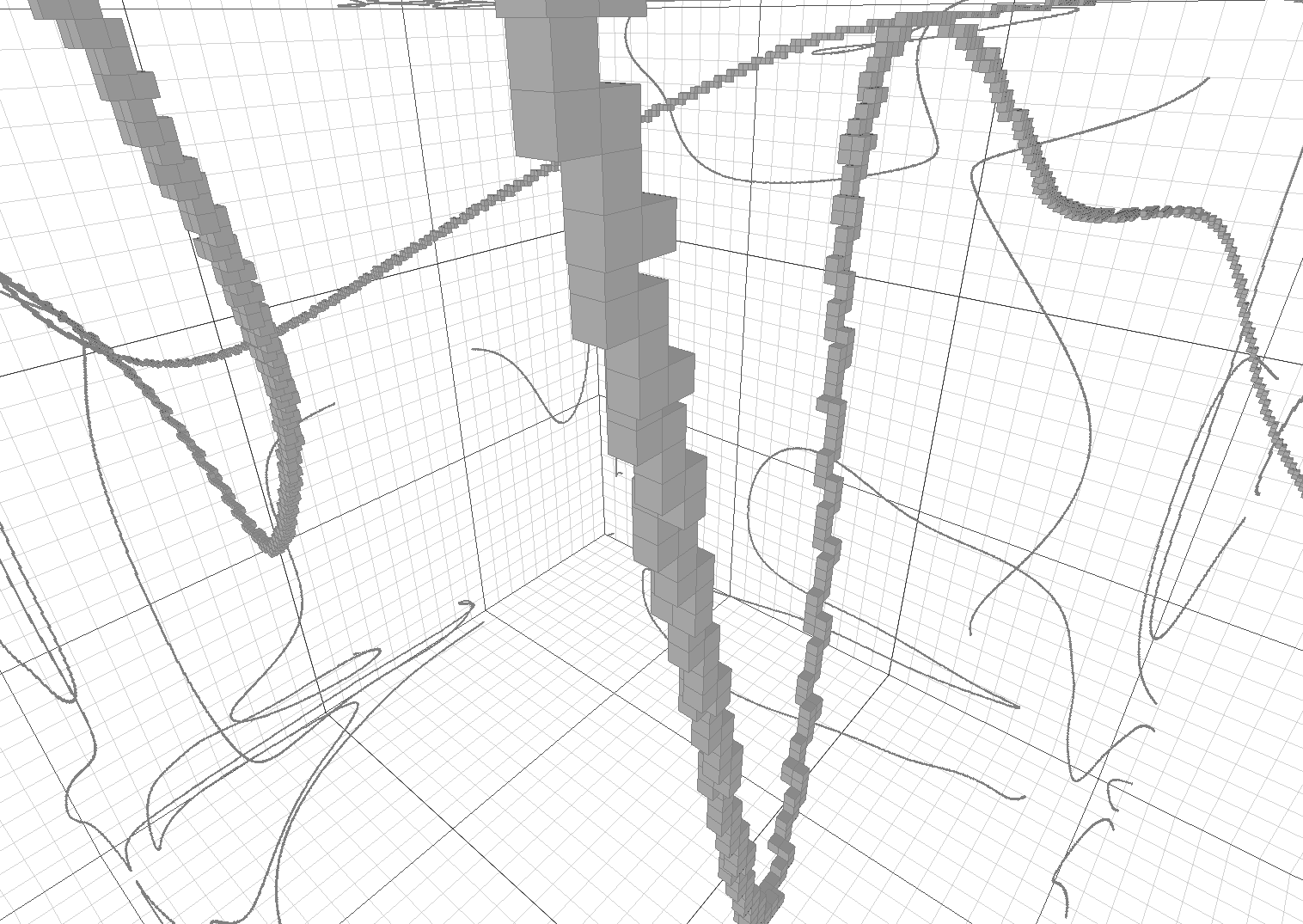}
  %\caption{Output of \texttt{voxelize} on a degree $100$ trivariate polynomial}
  \caption{Enclosing of a curve defined by $2$ trivariate polynomials
  of degree $100$}
  \label{fig:3D}
\end{minipage}
\end{figure}

% State of the art:
%  cite Plantinga-Vegter, Schneider, etc.
%  cite recent work from Yap et al. at ISSAC
%  cite Boissonat et al. with black-box evaluation

\subsection{Amortized Evaluation on a Grid of Boxes}

The first idea to reduce the evaluation redundancies is to use partial
evaluation. Assume that $F(x_1,x_2)$ is a
bivariate polynomial of degree $d$. Moreover, let
$(I_i)_{0 \leq i <n}$ and $(J_i)_{0\leq j <n}$ be two sequences of real
intervals. Using the H\"orner scheme, evaluating $F$ on a box requires
$O(d^2)$ arithmetic operations, and evaluating $F$ on all the boxes
$I_i \times J_j$ for $0 \leq i,j < n$ requires $O(d^2n^2)$ arithmetic
operations. By reorganizing the operations using partial evaluations,
the number of arithmetic operations can be reduced to
$O(dn(d+n))$. This idea is well known and was used for example to speed up the
multiplication of polynomials \cite{Pjsc94}. It is also currently
implemented in the well-spread library \texttt{NumPy} to evaluate
polynomials in $2$ and $3$ variables \cite{harris2020array}.

More precisely the operations are reordered as follows. For a given $I_i$,
the partial evaluation of $F$ in $I_i$ results in a univariate
polynomial $f_i$ of degree $d$. This step requires $O(d^2)$ arithmetic
operations. Then evaluating $f_i$ on $n$ intervals
requires $O(dn)$ arithmetic operations. Finally, repeating these
operations for all the $n$ intervals $I_i$, this allows us to
evaluate $F$ on all the boxes of the grid with a total
number of arithmetic operations in $O(dn(d+n))$. More generally, for
higher dimensions, this leads to the following result.
%, that is a direct
%corollary of the more general Theorem~\ref{thm:complexity}.

\begin{property}[\cite{Pjsc94}]
  \label{clm:dense}
  Let $F$ be a polynomial in $k$ variables and of degree at most $d-1$ in each variable. Let $X_1,
  \ldots,X_k$ be $k$ sets of $n$ real intervals each.
  Then it is possible to evaluate $F$ on all the boxes of $X_1 \times
  \cdots \times X_k$ in $O(kdn\max(n,d)^{k-1})$ arithmetic operations.
\end{property}

In the case where $n>d$, this approach results in a significant speedup
since the amortized number of arithmetic operations to evaluate $F$ on
each box of the grid is $O(kd)$ instead of $O(d^k)$.

\subsection{Amortized Evaluation on a Sparse Subset of a Grid}
For the simple subdivision algorithm mentioned at the beginning of the
introduction, if the
boxes created are never discarded, then each level of the subdivision
tree forms a dense grid of boxes. In this case, the partial evaluation
approach shown in the previous section can be applied directly to reduce
the total number of arithmetic operations required to evaluate $F$ on
each box with interval methods. In the general case though, many boxes
are discarded, and the boxes appearing in a given level of the
subdivision tree form a subset of a grid, as shown in
Figure~\ref{fig:2D}. The boxes created in the subdivision algorithm can
be handled in different orders. Using a breadth-first walk on the
subdivision, the boxes on the same level are a subset of a grid. In this
case, we need to evaluate a polynomial on a sparse subset of a grid.

To evaluate a polynomial on a general set of points, the case of a
univariate polynomial is well understood
\cite{Fstoc72,Hrr08,Mfocs22,IMissac23}. For multivariate polynomials,
there are fewer results that are efficient in practice when the points
are not arranged as a grid. A breakthrough, that was recently improved,
is a quasi-linear algorithm to evaluate a polynomial of degree $d$ in
$k$ variables on $d^k$ points in a finite field
\cite{Ustoc08,KUjc11,VLjc20,BGKMstoc22,BGGKUfocs22}. For multipoint
evaluation with real numbers, the only subquadratic algorithms are for
bivariate polynomials \cite{NZesa04}, or require precomputation more
than quadratic in the number of points \cite{VLjc21,VLjc23}. Finally, a
recent work addresses the case of approximate numerical evaluation
\cite{GHHKSfocs23}. Unfortunately, those approaches are not yet
efficient in practice. Our main result is a practical improvement to
amortize multipoint evaluations in the case were the points or boxes
that we consider are a sparse subset of a grid.
%The boxes created in the subdivision algorithm can be handled in
%different orders. Using a breadth-first walk on the subdivision, the
%boxes on the same level are a subset of a grid and

Boxes in a sparse subset of a grid can be gathered
and stored as a sparse tensor in the Compressed Sparse Fiber (CSF)
format \cite{SKiaaa15,CKApl18}. The CSF is a generalization of the
Compressed Row Format used to store the entries of a sparse matrix.
Then, $F$ can be evaluated efficiently on these boxes (see
Section~\ref{sec:evaluationcsf} for more details).
% using Theorem~\ref{thm:sparse}.
%the amortized evaluation
%scheme sketched in the previous sections.
This approach was implemented
in the library \texttt{voxelize}.  Figure~\ref{fig:3D} shows the output
boxes of the software \texttt{voxelize} enclosing an algebraic curve
defined by two polynomial equations of degree $100$, where the
coefficients are randomly drawn from a normal law centered at zero.
Performing the partial evaluation approach on a set of boxes in a CSF
format leads to Theorem~\ref{thm:sparse}.

%For the simple subdivision algorithm mentioned at the beginning, if the boxes
%created are never discarded, then each level of the subdivision tree
%form a dense grid of boxes. In this case, the partial evaluation approach shown in the previous
%section can be applied directly to reduce the total number of arithmetic
%operations required to evaluate $F$ on each boxes with interval methods.
%
%In the general case though, many boxes are discarded, and the boxes
%appearing in a given level of the subdivision tree forms a subset of a
%grid, as shown in Figure~\ref{fig:2D}. To handle the evaluation of $F$
%on a subset of a grid, it is possible to design a variant of the partial
%evaluation presented in the previous section. The key idea is to encode
%this sparse subset of boxes with a sparse tensor in the Compressed
%Sparse Fiber format or CSF \cite{SKiaaa15,CKApl18}.  The CSF is a generalization of the
%Compressed Row Format used to store the entries of a sparse matrix.
%Performing the partial evaluation approach on a set boxes in a CSF
%format leads to the following result.

\subsection{Notations}
For a set $E$, we denote by $|E|$ its number of elements. Then, we define the notations for the size of the projection of a
subset $E$ of a grid. In particular, the size of the projection is smaller when the
elements of $E$ are aligned within the grid.

%Given subset $E$ of a grid in dimension $k$, one indicator of how its
%elements are aligned within a grid is the size
%of the set is the size of its projection on the $i$ first coordinates or
%$i$ last coordinates.

\begin{definition}
  \label{def:size}
  Given a finite set $E \subset \mathbb N^k$, and an integer $i$ between
  $1$ and $k$, we denote by $N_i(E)$ (resp. $\widetilde N_i(E)$) the number of elements in the projection $E$ on
  the first (resp. last) $i$ coordinates, counting repeated projections
  only once. 
\end{definition}

Even though this definition holds for a set of integer tuples, it can be naturally extended for multivariate polynomials.
Indeed, for each monomial, we can associate its vector of exponents.
If $F$ is a polynomial in $k$ variables, for a given integer $i$, we can
define
$N_i(F)$ (resp. $\widetilde N_i(F)$) as the size of the projections of the set of vectors of
exponents of $F$ to their first (resp. last) $i$ coordinates.

For $S$ a set of points or boxes that is a subset of a grid, we can also
extend the definition of $N_i$ by simply
indexing the elements of $S$ by their integer positions in the grid. Letting $S_{ind}$
be the set of integer indices of the boxes of $S$, we can define
$N_i(S)$ by $N_i(S_{ind})$.

%TODO: add figure

\subsection{Main Result}
We can now state our main theorem to evaluate a multivariate
polynomial on a set of boxes that is a sparse subset of a grid of boxes
$G$ that is the Cartesian product of $k$ sets of intervals $X_1 \times
\cdots \times X_k$.

\begin{theorem}
  \label{thm:sparse}
  Let $F$ be a polynomial in $k$ variables, and $S$ be a subset of boxes of
  $G$. It is possible to evaluate $F$ on all the boxes
  of $S$ in $O(\sum_{i=0}^{k-1} \widetilde N_{k-i}(F) N_{i+1}(S))$
  arithmetic operations.
\end{theorem}

%\begin{corollary}
%  \label{cor:sparse}
%  Let $F$ be a polynomial in $k$ variables and of degree at most $d$ in each variable. Let $X_1,
%  \ldots,X_k \subset \mathbb R$ be $k$ finite sets of intervals,
%  and let $B$ be a subset of $X_1 \times \cdots \times X_k$. Finally,
%  let $N_i(B)$ be the number of boxes of the
%  projection of $B$ on the first $i$ coordinates, for $i$ from $1$ to
%  $k$.  Then it is possible to evaluate $F$ on all the points of $X_1 \times \cdots
%  X_k$ in $O(kd\max_{1\leq i \leq k}(d^{k-i}N_i(B)))$ arithmetic operations.
%\end{corollary}

When the set of boxes enclose a variety of dimension $j$,the
projection of $S$ on the first $j$ coordinates is often a dense grid. In
this case, we have the following corollary.

\begin{corollary}
\label{cor:variety}
%  When the set of boxes enclose a
%variety of dimension $j$,the
%projection of $S$ on the first $j$ coordinates is often a dense grid
%$X_1\times \cdots\times X_j$. If the sizes of $X_1,\ldots,X_j$ are
%larger than $d$ the degree of the polynomial $F$, then
%Theorem~\ref{thm:sparse} implies that we can evaluate each box of $S$ in
%$O(kd^{k-j+1})$ arithmetic operations on average, instead of $O(d^k)$
%operations.

For $1\leq j \leq k-1$, assume that the projection of $S$ on the first $j$ coordinates is:
\begin{itemize}
  \item[i.] a dense grid, denoted by $X_1\times\cdots\times X_j$
  \item[ii.] $|X_i| > d$ for all $1\leq i \leq j$.
\end{itemize}
%Theorem~\ref{thm:sparse} implies that
Then we can evaluate each box of $S$ in $O(j(d+1)^{k-j+1})$ arithmetic operations on average, instead of $O((d+1)^k)$ operations.
\end{corollary}
\begin{proof}[Corollary~\ref{cor:variety}]
First, if $F$ has degree at most $d$ in each variable, then
$\widetilde N_{k-i}(F)$ is less than $d^{k-i}$ for all non-negative
integers less or equal to $k$.

For $0\leq i < j$, Assumption $i$ implies that $N_{i+1}(S)
= N_i(S) |X_{i+1}|$. Then we deduce with Assumption $ii.$ that
$(d+1) N_i(S) \leq N_{i+1}(S)$. This implies that:
\begin{align*}
  \widetilde N_{k-i}(F) N_{i+1}(S) &\leq (d+1)^{k-i} N_{i+1}(S) \\
                                   &\leq (d+1)^{k+1-j}\\
                                   & N_j(S)\leq (d+1)^{k+1-j} |S|.
\end{align*}

For $i \geq j$ we have $\widetilde
N_{k-i}(F) \leq (d+1)^{k-i}$, such that $\widetilde N_{k-i}(F)
N_{i+1}(S) \leq (d+1)^{k-i} |S|$. Thus, the evaluation of $F$ on all
the boxes of $S$ is in $$O(j (d+1)^{k-j+1}|S| +
\sum_{i=j}^{k-1}(d+1)^{k-i}|S|) = O(j(d+1)^{k-j+1}).$$ In particular, the amortized
cost of evaluating each box is in $O(jd^{k-j+1})$ arithmetic operations
instead of $O(d^k)$ with a direct algorithm.
\end{proof}

%As shown in the next sections, this data structure is particularly well
%suited when combined with subdivision algorithms.

%As a consequence, if each $X_i$ has more than $d$ intervals, and if the
%projection of $A$ on the $i$ first coordinates is $X_1\times\cdots\times
%X_i$, then the amortized cost is in $O(kd^{k+1-i})$ instead of $O(d^k)$.
%As shown in the next section, this data structure is particularly well
%suited when combined with subdivision algorithms.

%\subsection*{Amortized subdivision algorithm}
%The boxes created in the subdivision algorithm can be handled in
%different orders. Using a breadth-first walk on the subdivision, the
%boxes on the same level can be gathered and stored as a sparse tensor
%in the CSF format. Then $F$ can be evaluated on them using the amortized
%evaluation scheme sketched in the previous sections.
%
%This approach was implemented in the library \texttt{voxelize}.
%Figure~\ref{fig:3D} shows the output boxes of the software
%\texttt{voxelize} enclosing an algebraic curve defined by two
%polynomial equations of degree $100$, where the coefficients are
%randomly drawn from a normal law centered at zero.

\section{Evaluating Polynomials with Compressed Sparse Fibers}
\label{sec:evaluation}

\subsection{Sparse Tensor Data Structure}
\label{sec:csf}
The main data structure used in our algorithms is the Compressed Sparse
Fiber, as described in \cite{SKiaaa15,CKApl18}. This data structure is
well suited to store a subset of a grid in high dimension. It can be
seen as a generalization of the classical Compressed Sparse Row data structure used to store
the entries of a sparse matrix as in Figure~\ref{fig:dcsr}.

\begin{figure}
  \centering
  \includegraphics[width=0.6\linewidth]{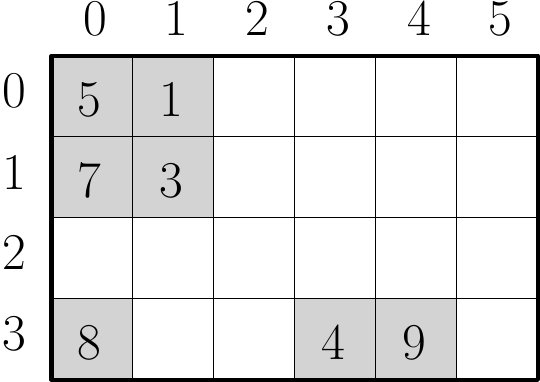}
  \caption{Numbers stored in a sparse matrix}
  \label{fig:dcsr}
\end{figure}

For a subset of a 2D grid, the data structure is a labeled tree that
stores the positions of
the non-empty rows in the children of the root node, and then in each row, the position of the
non-empty entries are stored in the children of the corresponding node (Figure~\ref{fig:graph}). In higher dimension $k$, this idea is applied
recursively. Let $E$ be a subset of points in $\mathbb N^k$. For
$t=(t_1,\ldots,t_{\ell}) \subset \mathbb N^\ell$ a
tuple of size $\ell<k$, we denote by $\pi_t(E)$ the subset of $\mathbb
N$ defined by:
$$\pi_t(E) = \{i \in \mathbb N \mid \exists y_{\ell+2},\ldots y_k \in
  \mathbb N \text{ such that }
(t_1,\ldots,t_\ell,i,y_{\ell+2},\ldots,y_k) \in E\}.$$

Then the Compressed Sparse Fiber (or CSF) data structure associated to $E$
is a labeled tree of depth $k$ defined recursively as
follows. The root of the tree is at depth $0$ and its children are the nodes labeled by
the elements of $\pi_{\emptyset}(E)$, where $\emptyset$ denotes the
empty tuple. Consider now a node $N$ of the tree at depth $1 \leq \ell
< k$. Let $t(N)$ be the tuple of size $\ell$, where the $i$-th
coordinate is the label of the $i$-th node on the path from the root to
$N$. Then the children of $N$ are the nodes labeled by the elements of
$\pi_{t(N)}(E)$. Finally, for a node $N$ at depth $k$, it is possible to add a
leaf that can be labeled with the value of the entry associated to the
tuple $t(N)$. Given a CSF data structure, the corresponding set of tuple
$E$ is unique and is called its \emph{support}.

%given by $k$ lists of sets $L_0, \ldots, L_{k-1}$, defined as follows:
%\begin{align}
%  L_0 &= [\pi_{()}(E)]\\
%  L_{\ell+1} &= [\pi_t(E) \mid 
%
%\end{align}

As an example, using the compressed sparse data structure to store the
sparse matrix given in Figure~\ref{fig:dcsr}, we get the tree shown in
Figure~\ref{fig:graph}, and its support is $\{(0,0), (0,1), (1,0),
(1,1), (3,0), (3,3), (3,4) \}$.

\begin{figure}
  \centering
  \includegraphics[width=0.6\linewidth]{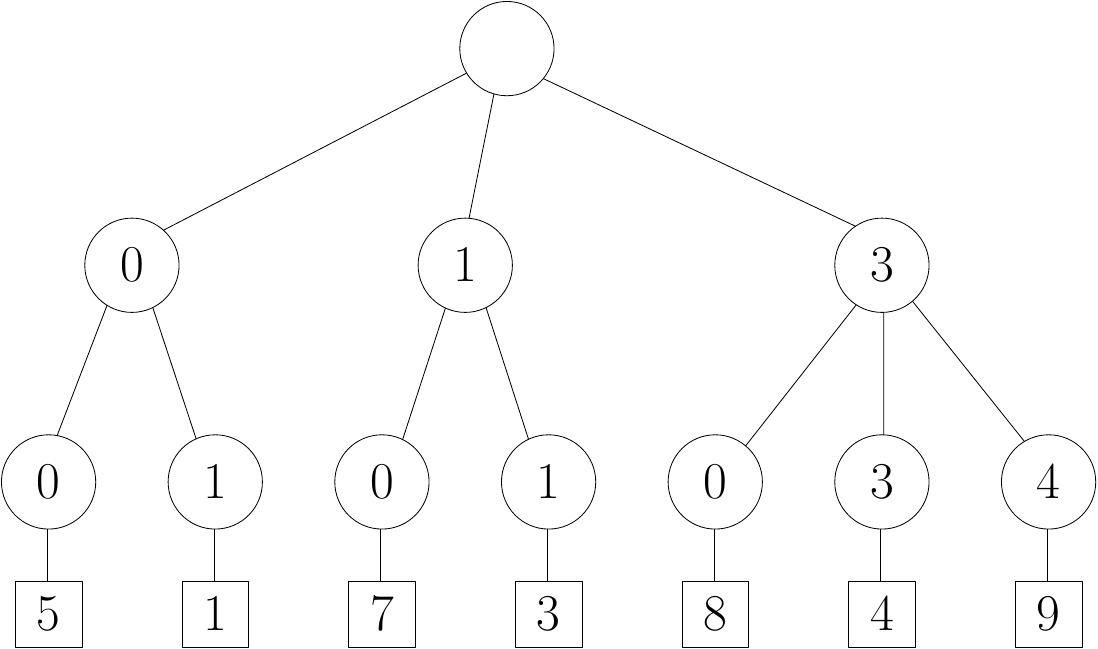}
\caption{Compressed Sparse Fiber associated to the sparse matrix}
  \label{fig:graph}
\end{figure}

%Finally, given a set $E$ and a depth $i$, we can define the size of a
%projection as follows.
%
%\begin{definition}
%  \label{def:size}
%  Given a finite set $E \subset \mathbb N^k$, and an integer $i$ between
%  $1$ and $k$, we denote by $N_i(E)$ the number of elements in the projection $E$ on
%  the first $i$ coordinates.
%\end{definition}
\begin{remark}
  \label{rem:size}
  Given $T$ a CSF data structure associated to $E$, remark the size of the
  projection on the first $i$ coordinates is the number of nodes of
  depth $i$ in $T$. In particular, we have $N_i(E) = \sum_{t\in\mathbb N^{i-1}} |\pi_t(E)|$.
  %Moreover, we have $N_k(E) = |E|$.
\end{remark}

\subsubsection{Representing a Multivariate Polynomial.}

A natural application of the Compressed Sparse Fiber data structure is
to encode the monomials of a sparse polynomial. Given a polynomial in
$k$ variables, the exponents of each monomial can be represented as a
$k$-tuple of integers in $\mathbb N^k$, and the coefficients can be
represented as the entries associated to each tuple. Using this
representation, it can be directly encoded in a CSF data structure.
Given a polynomial $F(x_1,\ldots,x_k)$, we denote by $T_F$ the CSF tree associated to
$F$.  By extension of Definition~\ref{def:size}, we define the size
$N_i(F)$ as the size $N_i(T_f)$ of its corresponding CSF tree truncated
to depth $i$.

For example the following polynomial in two variables would be encoded
with the CSF tree in Figure~\ref{fig:graph}:
$$5 + x_2 + 7x_1 + 3x_1x_2 + 8x_1^3 + 4x_1^3x_2^3 + 9 x_1^3x_2^4$$

\subsubsection{Representing a Set of Boxes.}

A sparse subset of a grid of boxes can also be represented with the
Compressed Sparse Fiber data structure, by applying it to the indexes of
the boxes within the grid.  Without loss of generality, consider
a subdivision of the unit box $[0,1]^k$ into $n^k$ smaller boxes, where
each box is a product of intervals of the form $\prod_{i=1}^k [a_i,
b_i]$, where $a_i$ and $b_i$ are real numbers. In the case where the
subdivision is uniform, let $G_n$ be the set of
these $n^k$ boxes. Each cube from $G_n$ can be indexed by a $k$-tuple of
integers in $\mathbb N^k$. In particular, for a sparse subset of $G_n$, we
can associate the set $B$ of the indices of its boxes.
%Let $G$ be the set of these $n^k$ cubes, and let $E$ be a subset of $G$.
Then, $B$ can be encoded in a CSF data structure. In this case, the tree
we construct won't have leaves since there is no entry associated to
each box.

\subsection{Evaluation Algorithm}
\label{sec:evaluationcsf}
\subsubsection{One Variable.}
A classical way to evaluate a univariate polynomial on a point is the
H\"orner algorithm that we recall in Algorithm~\ref{alg:horner} for the
evaluation of a sparse polynomial on an interval.

\begin{algorithm}
  \DontPrintSemicolon
  \KwInput{An interval $I$ and a polynomial $F(x) = a_0x^{e_0} + \cdots
    + a_\ell x^{e_\ell}$ where:
    \begin{tabular}{rl}
      $e_0 < \cdots < e_l$ & are integers\\
      $a_0,\ldots,a_\ell$ & are real numbers or intervals.
    \end{tabular}
  }
  \KwOutput{
    The interval obtained by evaluating $F$ on $I$ with the
    H\"orner scheme.
  }
  \vspace{1em}
  $J \gets a_l$\;
  \For{$j$ from $\ell-1$ to $0$}{
    $J \gets J \times I^{e_{j+1}-e_j} + a_j$ \;
  }
  \Return J\;
  \caption{H\"orner algorithm}
  \label{alg:horner}
\end{algorithm}

\subsubsection{Several Variables.}
For multivariate polynomials $F$, we can use the H\"orner scheme recursively.
Moreover, if we want to evaluate $F$ on a set of boxes,
Algorithm~\ref{alg:evaluationcsf} generalizes the H\"orner
scheme to the case where $F$ and the boxes are stored in a CSF data
structure. The key idea in Algorithm~\ref{alg:evaluationcsf} is that for
boxes that share the same coordinate, we only evaluate the polynomial
partially on those coordinates. Then we reuse those partially evaluated
polynomials to evaluate the boxes on the remaining coordinates.

\begin{algorithm}
  \DontPrintSemicolon
  \KwInput{%
    \begin{tabular}[t]{rl}
      $F$ & a polynomial in $k$ variables\\
      $T$ & a CSF tree representing the indices of a subset $S$ of boxes
      of a grid\\
          & $X_1 \times\cdots\times X_k$, where $X_i$ is a set of intervals.
    \end{tabular}
  }
  \KwOutput{                               
    A CSF data structure representing the evaluation of the polynomial
    represented by $F$ on all the boxes in $S$.
  }
  \vspace{1em}
  \SetKwFunction{EvaluationCSF}{EvaluationCSF}
  \SetKwProg{Function}{Function}{:}{}
  \Function{\EvaluationCSF{$F$,$T$}}{
    $X_1 \gets$ the list of intervals of the first coordinate in the grid $G$\;
    $L \gets$ empty list \;
    \For{i in $\pi_{\emptyset}(T)$}{
      $I \gets$ the interval of index $i$ in $X_1$\;
      $F_I \gets F(I, x_2,\ldots,x_n)$\;
      \If{F is univariate}{
        Append $F_I$ to $L$\;
      }
      \Else{
      $T_i \gets$ the subtree of $T$ rooted at the node at depth $1$ with label $i$ \;
      $L_i \gets \EvaluationCSF(F_I,T_i)$\;
      Append $L_i$ to $L$\;
      }
    }
    \Return $L$\;
  }
  \caption{Evaluation on a set of boxes}
  \label{alg:evaluationcsf}
\end{algorithm}

%\subsection{Complexity analysis}
The advantage of using the approach in Algorithm~\ref{alg:evaluationcsf}
is that it allows us to amortize the cost of the evaluation
when several boxes have the same projection. In the following,
% we assume
%that $G$ is a grid of boxes in dimension $k$, that is the Cartesian
%products of $k$ sets of intervals.
we will prove that the complexity of
Algorithm~\ref{alg:evaluationcsf} is in $O(\sum_{i=0}^{k-1} \widetilde N_{k-i}(F)
N_{i+1}(S))$ arithmetic operations, which will prove
Theorem~\ref{thm:sparse}.

\begin{proof}[Theorem~\ref{thm:sparse}]
  Since Algorithm~\ref{alg:evaluationcsf} is recursive, we will prove
  its complexity by recurrence. Algorithm~\ref{alg:evaluationcsf} is a loop over the nodes of
  the root of $T$. In particular, this loop will be called $N_1(S)$. In each loop, the dominating complexities are in
  line~6 and~11. In line~6, the complexity of evaluating partially $F$
  in one variable $x_1$ is $\widetilde N_{k}(F)$. Thus, the total
  complexity carried by line~6 is in $O(\widetilde N_k(F)N_1(S))$. And if $F$ is
  univariate, the complexity of Algorithm~\ref{alg:evaluationcsf} is in $O(\widetilde N_1(F) N_1(S))$.

  Then, if $F$ is a polynomial in $k$ variables with $k>1$, the number
  of operations is again carried by lines~6 and~11. Let $S_I$ be the set of
  boxes represented by the tree $T_I$. By recurrence the number of
  operations in line~11 is in $$O\left(\sum_{i=0}^{k-2}
  \widetilde N_{k-1-i}(F_I) N_{i+1}(S_I)\right).$$ In particular, remark that
  $\widetilde N_{k-1-i}(F_I) = \widetilde N_{k-1-i}(F)$. And using
  Remark~\ref{rem:size}, the sum of
  the $N_{i+1}(S_I)$ on all the intervals $I$ children of the root of
  $T$ is equal to $N_{i+2}(S)$. Thus, the complexity of
  Algorithm~\ref{alg:evaluationcsf} carried by line~11 is
  $O\left(\sum_{i=0}^{k-2} \widetilde
  N_{k-1-i}(F)N_{i+2}(S)\right)$. By changing the index of the sum, this
  complexity becomes $$O\left(\sum_{i=1}^{k-1} \widetilde
  N_{k-i}(F)N_{i+1}(S)\right).$$ Since the
  complexity carried by line~6 is $O(\widetilde
  N_k(F)N_1(S))$, this concludes the proof.

\end{proof}

%\begin{theorem}
%  \label{thm:complexity}
%  Let $F$ be a polynomial in $k$ variables, and $B$ be a subset of boxes of
%  $G$. Algorithm~\ref{alg:evaluationcsf} evaluates $F$ on all the boxes
%  of $B$ in $O(T(F,B))$ arithmetic operations, where:
%  $$T(F,B) = \sum_{i=0}^{k-1} N_{k-i}(F) N_{i+1}(B)$$
%\end{theorem}

\section{Applications}
\subsection{The Fast Fourier Transform Revisited}
\label{sec:fft}

Given a vector $u$ of $d+1$ complex numbers $u_0, \ldots, u_d$, its discrete
Fourier Transform is the vector $v$ of $d+1$ complex numbers $v_0,\ldots,v_d$
such that:
\begin{equation}
v_k = \sum_{j=0}^d u_j e^{-i2\pi \frac k {d+1}  j}
\label{eq:fft}
\end{equation}

The fast Fourier Transform algorithm returns the vector
$v$ using $O(d\log d)$ arithmetic operations. If we reinterpret
Equation~\eqref{eq:fft} as the evaluation of a multivariate polynomial
on a set of points stored with a CSF tree data structure, we can use
Algorithm~\ref{alg:evaluationcsf} to compute the discrete Fourier
transform in $O(d\log d)$ arithmetic operations.

Without restriction of generality, assume that there exists an integer
$k$ such that $d+1=2^k$ is a power of two. Let $F$ be the polynomial in
$k$ variables defined by:
$$F = \sum_{(i_1,\ldots,i_{k}) \in \{0,1\}^k}
u_{{}_{i_1+\cdots+i_{k}2^{k-1}}} x_1^{i_1}\cdots x_{k}^{i_{k}}$$

Moreover, let $w$ be the $(d+1)$-th root of unity $e^{-i2\pi/(d+1)}$.
For $1\leq j\leq k$, let $X_j = \{1, w^{2^{k-j}}\}$, and let $G$ be the grid of points $g_{i_1,\ldots,i_{k}}$ in $\mathbb
C^k$ for $(i_1,\ldots,i_{k}) \in \{0,1\}^k$, defined by:
$$g_{i_1,\ldots,i_{k}} = (w^{i_12^{k-1}}, \ldots, w^{i_{k}}) \in
X_1\times\cdots\times X_k$$

Then, using the notations of Equations~\eqref{eq:fft}, for an integer $j
= i_12^{k-1} + \cdots + i_{k}$ we have $v_j =
F(g_{i_1,\ldots,i_{k}})$.  The polynomial $F$ has a degree at most $1$
in each variable and the set of points on which $F$ is evaluated is a
Cartesian product $X_1\times \cdots \times X_k$ where $X_j$ has size $2$
for all $1 \leq j \leq k$. Then, using Claim~\ref{clm:dense}, this
evaluation can be done using $O(k2^k)$, that is $O(d\log d)$ arithmetic
operations.

\subsection{Subdivision Algorithm}
\label{sec:subdivision}
A classical approach to find the zero locus of a set of a polynomial equation is to
use a subdivision algorithm. Given a polynomial equation $F$ and a
box $B$, assume that we have two criteria $C_0(F,B)$ and $C_1(F,B)$ such
that:
\begin{itemize}
  \item if $C_0(F,B)$ is true, then $F$ doesn't vanish in $B$
  \item if $C_1(F,B)$ is true, then $F$ vanishes in $B$
\end{itemize}
                                                                
The idea of a subdivision algorithm is to start with a set of boxes,
and to bisect them recursively until the criterion $C_0$ is true, or $C_1$ is
true and the size is smaller than a given threshold. Recall that the
grid $G_n$ is the set of $n^k$ boxes obtained by subdividing uniformly $[0,1]^k$
in $n$ boxes in all the directions. Given a box $B$
from the grid $G_n$, if we bisect it uniformly in $2$ in all the
directions, we end up with a set of $2^k$ boxes, all of them included in
$G_{2n}$. In particular, if we bisect a set of boxes in $G_{n}$, we end
up with a set of boxes in $G_{2n}$. Moreover, if the criteria $C_0$ and
$C_1$ are based on polynomial evaluations, we can use
Algorithm~\ref{alg:evaluationcsf} to amortize the evaluation. This leads
to Algorithm~\ref{alg:subdivision}, that computes a set of boxes that
enclose the zero-set of a polynomial equation. If we want to
compute the zero set of a system of polynomial equations and
inequalities, Algorithm~\ref{alg:subdivision} can be used unchanged, and
the criteria $C_1$ and $C_0$ can be easily adapted to detect if a
system of equations has solutions or not in a given box. To ensure that
Algorithm~\ref{alg:subdivision} terminates, it is necessary that for
boxes small enough, either criterion $C_0$ or $C_1$ succeed.

%In order to evaluate efficiently $F$ on the sets of boxes created during
%a subdivision algorithm, we add the constraint that the set of boxes on
%which we evaluate $F$ are a subset of a grid. TODO
%
%If the bisection step cuts a box $B$ in $\mathbb R^k$ uniformly in all
%the directions, in $2^k$ boxes 

\begin{algorithm}
  \DontPrintSemicolon
  \KwInput{%
    \begin{tabular}[t]{rl}
      $F$ & a multivariate polynomial\\
      $\varepsilon$ & a positive threshold real number
    \end{tabular}
  }
  \KwOutput{                               
    A CSF data structure representing the boxes of size at most
    $\varepsilon$, such that $F$ vanishes in all the boxes, and
    doesn't vanish outside the boxes.
  }
  $S \gets \{ [0,1]^k \}$ \;
  $R \gets \{ \}$ \;
  size $\gets 1$ \;
  
  \While{$S$ is not empty}{
    $S \gets$ set of boxes $B$ in $S$ not satisfying $C_0(F,B)$ \;
    \If{size $< \varepsilon$}{
      $R \gets$ $R$ union the set of boxes $B$ in $S$ satisfying $C_1(F,B)$ \;
      $S \gets$ set of boxes $B$ in $S$ not satisfying $C_1(F,B)$ \;
    }
    $S \gets$ set of boxes bisected from the boxes in $S$ \;
    size $\gets$ size$/2$ \;
  }
  \Return $R$ \;

  \caption{Simple subdivision algorithm to enclose the zero locus of a
  polynomial equation}
  \label{alg:subdivision}
\end{algorithm}

\subsubsection{Criteria for Exclusion and Inclusion}

\paragraph{Exclusion Criterion.}
A simple exclusion criterion $C_0(F,B)$ consists in evaluating $F$ on
$B$ using interval arithmetic.  Interval arithmetic is the
generalization of standard arithmetic operations to the case where
numbers are replaced by intervals. If $[a,b]$ and $[c,d]$ are two
intervals, the result of $[a,b] + [c,d]$ is the interval $[a+c, b+d]$.
If $F$ is a polynomial in $k$ variables and $B$ is a product of $k$
intervals, we denote by $\Box F(B)$ the interval returned when $F$ is evaluated on $B$ using
interval arithmetic. The main property of interval arithmetic is that
the interval $\Box F(B)$ satisfies $\{F(x) \mid x \in B\} \subset \Box
F(B)$. In particular, if $0 \notin \Box F(B)$, then $F$ does not vanish
in $B$. Thus, we can define $C_0(F, B)$ as the predicate $0 \notin \Box
F(B)$.

The exclusion criterion can also be computed using other schemes to
evaluate $F$ on $B$, such as the Taylor form, which can reduce the
overesetimation near the zeros of $F$ \cite[\S 3.5]{Hhal15}.

\begin{definition}[Taylor Form~{\cite[Definition
3.3]{RRca84},\cite{HKYissac21}}]
  \label{def:taylor}
  If $c$ is the middle point of $B$, for a given integer $m$, the Taylor form of order
$m$ of the polynomial $F$ in $k$ variables is defined by:
$$T_m(F,x) = F(c) + \cdots + \frac{F^{(m-1)}(c)}{(m-1)!} (x-c)^{m-1} +
\frac{\Box F^{(m)}(B)}{m!} (x-c)^m$$
where $x=(x_1,\ldots,x_k)$ is a tuple of symbolic variables.

\end{definition}

This evaluation scheme satisfies the property $\{F(x) \mid x\in B\}
\subset T_m(F,B)$, such that $0\notin T_m(F,B)$ implies that $F$ does
not vanish in $B$. In the case of a system of several equations, we can
simply test if any of the input polynomial does not contain $0$.

\paragraph{Inclusion Criterion.}
For the inclusion criterion $C_1(F,B)$ to detect if $F$ vanishes in $B$,
a simple test consists in evaluating $F$ on all the vertices of $B$ and
returning \texttt{true} if two of them have different signs, and
\texttt{false} if all the signs are the same. Remark that the set of all
the vertices of all the boxes are a subset of a grid, and thus we can
also use Algorithm~\ref{alg:evaluationcsf} to amortize the cost of
their evaluation.

The inclusion criterion $C_1(F,B)$ can also be based on the Taylor form if we computed it
with order $m$, where $m$ is an integer greater or equal to $2$. Let
$\ell(x)$ be the linear
part of $T_m(F,x)$. Let $V_{min}$ be a vertex of $B$ that minimizes $\ell$ and
$V_{max}$ one that maximizes $\ell$. Then we can reduce the evaluation
of $F$ to the vertices $V_{min}$ and $V_{max}$. We can also use the
Taylor form to evaluate lower and upper bounds of the values of $F$ at
$V_{min}$ and $V_{max}$. In this case, our predicate will return
\texttt{true} if the lower bound on $F(V_{max})$ is positive and
the upper bound on $F(V_{min})$ is negative.

If $F$ is a vector of multiple polynomials, and if we want to test if they vanish
simultaneously inside a box, we can use a criterion $C_1$ derived from the Newton
Interval criterion~\cite{Nbook90,Gc07}.
%No such test has been implemented in
%$\texttt{voxelize}$, yet based on state-of-the-art methods, we sketch
%here a test that could be used for the inclusion criterion $C_1$ for
%multiple polynomial equations.
First, when the number of input equations
$F_1=0, \ldots, F_k=0$ is equal to the number of variables, the Newton
Interval criterion can be seen as a fixed-point theorem. Letting $S$ be the $k
\times k$ matrix defined by

$$S_{ij} = \begin{cases}
  \frac{F_i(x_1,\ldots,x_{j-1},x_{j},c_{j+1},\ldots,c_k) -
  F_i(x_1,\ldots,x_{j-1},c_{j},c_{j+1},\ldots,c_k)}{x_{j}-c_{j}} &\text{if } x_j \neq c_j\\
        \frac{dF}{dx_j}(x_1,\ldots,x_{j-1},c_j,\ldots,c_k) & \text{if } x_j=c_j
      \end{cases}
  ,$$
and $c$ be the center of the box $B$, we define the formula $N(x) =
c-S(x)^{-1} F(c)$. If $N(B) \subset B$, the fixed-point theorem
ensures that there exists a point $x_0$ in $B$ such that $N(x_0)=x_0$,
which is equivalent to $F_1(x_0)=0,\ldots,F_k(x_0)=0$. Otherwise, when
the number of polynomial equations is less than the number of variables,
we can intersect the box with the linear space spanned by the gradient
vectors of the input polynomial at the center of the box $B$. Then we
can use the Newton Interval criterion on the resulting system that has
as many equations as variables.

\section{Experiments}
\label{sec:experiments}

Algorithm~\ref{alg:evaluationcsf} and~\ref{alg:subdivision} have been
implemented in \texttt{C++} in the software \texttt{voxelize}. This
software can take as input a list of polynomial equations and polynomial
inequalities, and it returns a list of boxes enclosing the set of points
where the input system has solutions. Furthermore, if the input is a single polynomial
equation, then it is guaranteed to vanish in each box returned by
\texttt{voxelize} that are larger than a threshold given by the user.
The software can be used as a standalone program, taking one file per
polynomial, or it can be used through a python interface.

The criterion $C_0$ used to exclude boxes is based on the Taylor form
evaluation scheme described in Definition~\ref{def:taylor}. The
criterion $C_1$ is implemented in the case where the input is a single
polynomial equation, and it follows the approach based on the Taylor
form detailed at the end of Section~\ref{sec:subdivision}. In the case
of multiple input polynomial equations, the subdivision process stops
when the boxes are smaller than a threshold given by the user.

\subsection{Random Polynomials}
\label{sec:random}
In Table~\ref{tab:random}, we show the time to enclose the zero-set of
polynomial equations in $k$ variables where $k$ is either $2,3$ or $4$.
In each case, we consider three cases:
  a hypersurface defined by one equation, a curve defined by $k-1$
  equations, points defined by $k$ equations. And for each case, we
  generated random polynomials of total degree either $20$ or $100$,
  except for $k=4$ where \texttt{voxelize} could not handle polynomials
  in $4$ variables and total degree $100$. The random coefficients are
  floating-point numbers with double precision uniformly sampled between
  $-10$ and $10$.

The computation have been done on a laptop with a $1.9$GHz CPU and $16$G
of RAM. The tests have been done with one thread, for easier comparison
with other single-thread programs. Note that \texttt{voxelize} is also
implemented with the multi-thread library \texttt{openmp} and it can
distribute the computations on several threads. Up to our knowledge,
\texttt{voxelize} is the only available software that can handle the
systems with polynomials of degree $100$ in $3$ variables presented
in Table~\ref{tab:random}.

\begin{table}
  \centering
  \caption{Timing in seconds for computing enclosing boxes in the cube
  $[-2,2]^k$. For points and curves, the subdivsion process
  stopped for boxes smaller than ${2^{-8}\simeq0.004}$. For
hypersurfaces, the subdivision process stopped when either the criterion
$C_0$ or $C_1$ was satisfied on all the boxes, and the boxes had a size
smaller than $2^{-5} \simeq 0.03$.}
  \label{tab:random}
  \begin{threeparttable}
  \begin{tabular}{|l|r@{.}l@{\quad}r@{.}l|r@{.}l@{\quad}r@{}l|r@{}l|}
  \hline
  \textbf{dimension $k$}& \multicolumn{4}{c|}{\textbf{2D}} &
  \multicolumn{4}{c|}{\textbf{3D}} & \multicolumn{2}{c|}{\textbf{4D}} \\ \hline
  \textbf{degree $d$}& \multicolumn{2}{c@{\quad}}{\textbf{20}} &
  \multicolumn{2}{c|}{\textbf{100}} &
  \multicolumn{2}{c@{\quad}}{\textbf{20}} &
  \multicolumn{2}{c|}{\textbf{100}} & \multicolumn{2}{c|}{\textbf{20}} \\ \hline
  points\tnote{2}\hfill ($k$ equations) & 0&006 & 0&32 & 0&5 & 273& & 56& \\
  curves\tnote{2}\hfill ($k-1$ equations) & 0&062 & 0&31 & 1&3 & 270& & 91& \\
  hypersurfaces ($1$ equation) & 0&062 & 0&31 & 1&1 & 412& & 373& \\ \hline
\end{tabular}
  \begin{tablenotes}
    \item[2] Only the exclusion criterion was implemented
      for this case, and not the inclusion criterion.
  \end{tablenotes}
  \end{threeparttable}
\end{table}

\subsection{Polynomials Coming from Applications}
\label{sec:applications}
We also used the software on two polynomial systems coming from robotics
and automatic applications. In these cases, we compared our software
with the state-of-the-art subdivision software
\texttt{ibex}. The \texttt{ibex} software is a general subdivision software including a specific feature
called contractors~\cite{CJai09}. A contractor is an operator that takes as input a
function $F$ and a box $B$, and that returns a smaller box $B'$ such
that the intersection of $B'$ with the zero set $Z$ of $F$ is the same as
the intersection of $B$ with $Z$.

\paragraph{Robotics.}
In robotics, a classical problem is to compute the parallel singularities of a
robot. That is the set of control parameters around which the robot can
be assembled in two nearby configurations. In particular, the
following set of equations defines the singularities in the orientation
space of the $3$-\underline{PP}PS manipulator~\cite{CGCCMAjmr12}. The
orientation space is modeled with $4$ quaternion variables, commonly used to
parametrize the rotation matrices in $3D$. The sum of the squares of the
quaternion variables is constrained to be $1$.

$$(R) \begin{cases}
0\quad=&\begin{array}[t]{@{}l@{}}
-6\,{Q_{{2}}}^{2}Q_{{3}}Q_{{1}}+6\,Q_{{2}}{Q_{{3}}}^{2}Q_{{4}}+3\,
\sqrt {3}{Q_{{2}}}^{2}Q_{{3}}Q_{{4}}
\\
\hspace{3ex} {} -6\,Q_{{2}}{Q_{{1}}}^{2}Q_{{4
}}+6\,Q_{{1}}{Q_{{4}}}^{2}Q_{{3
}}\\
\hspace{3ex} {} -3\,\sqrt {3}Q_{{2}}Q_{{1}}{Q_{{4}}}^{2}+3\,\sqrt {3}Q_{{2}}{Q_{{3}}}^{2}Q_{{1}}
\\
\hspace{3ex} {} -3\,\sqrt {3}Q_{{3}}{Q_{{1}}}^{2}Q_{{4}}+\sqrt {3}{Q_{{2}}}^{3}Q_{{1}}
\\
\hspace{3ex} {} -\sqrt {3}Q_{{2}}{Q_{{1}}}^{3}+Q_{{4}}\sqrt {3}{Q_{{3}}}^{3}-Q
_{{3}}\sqrt {3}{Q_{{4}}}^{3}
\end{array}\\

1 \quad=&Q_1^2+Q_2^2+Q_3^2+Q_4^2
\end{cases}$$

\paragraph{Automatic.}

In control theory, a common problem is to decide if it is possible to
add a controller to a dynamic system such that it becomes stable. In
some case, this problem can be reduced to decide if a polynomial system
does not vanish on complex numbers of modulus less than one. For
example, the following system in $3$ complex variables was communicated by Thomas
Cluzeau and Alban Quadrat. If it has no solution where $z_1, z_2$ and
$z_3$ have a modulus less than $1$, then it is possible to design a
stable controller for the corresponding dynamic system. 

$$(A) \begin{cases}
  &|z_1| \leq 1\\
  &|z_2| \leq 1\\
  &|z_3| \leq 1\\
  &0=z_{1} z_{2}^{2}-z_{1} z_{3}-2\\

  &0=\begin{array}[t]{@{}l@{}}
    12 z_{2}^{3} z_{3}^{3}-2 z_{1}^{2} z_{2} z_{3}^{2}+z_{2}^{3} z_{3}^{2} -2 z_{2}^{2} z_{3}^{3}-12 z_{2} z_{3}^{4}+2 z_{1}^{2} z_{3}^{2}-z_{2} z_{3}^{3}\\
    -2 z_{1} z_{2} z_{3}-7 z_{2}^{3}-10 z_{2}^{2} z_{3}+14 z_{1} z_{3}-8 z_{2}^{2}+9 z_{2} z_{3}+12 z_{3}^{2}+30 z_{3}+2
     \end{array}\\
  &0=\begin{array}[t]{@{}l@{}}
      z_{1}^{3} z_{3}^{3}+z_{1} z_{2} z_{3}^{4}-z_{1}^{3} z_{3}^{2}+z_{1}
  z_{3}^{4}-12 z_{2}^{2} z_{3}^{3}-6 z_{1}^{2} z_{2} z_{3}+3 z_{1}^{2}
  z_{3}^{2}-z_{1} z_{2} z_{3}^{2}\\
  -z_{2}^{2} z_{3}^{2}-10 z_{2}
  z_{3}^{3}-7 z_{1}^{2} z_{3}-12 z_{1} z_{2} z_{3}-2 z_{1}
  z_{3}^{2}-z_{2} z_{3}^{2}+2 z_{3}^{3}-z_{1} z_{2}\\
  -9 z_{1} z_{3}+7
  z_{2}^{2}+10 z_{2} z_{3}-z_{1}+15 z_{2}+8 z_{3}+8
  \end{array}\\

  &0=\begin{array}[t]{@{}l@{}}
  z_{1}^{3} z_{2} z_{3}^{2}-z_{1}^{3} z_{3}^{2}+z_{1}
  z_{3}^{4}+z_{1}^{2} z_{2} z_{3}-12 z_{2} z_{3}^{3}-7 z_{1}^{2}
  z_{3}-z_{1} z_{2} z_{3}-z_{1} z_{3}^{2}\\
  -z_{2} z_{3}^{2}+2 z_{3}^{3}-11
  z_{1} z_{3}-z_{1}+7 z_{2}+10 z_{3}+8
\end{array}\\
\end{cases}$$

By using the change of variable $z_j = x_j + i y_j$, we get $8$
polynomial equations in $6$ variables, with the additional inequalities
$x_i^2+y_1^2 \leq 1$.

\paragraph{Experiences.}

We used \texttt{voxelize} and \texttt{ibex} on those two system of
polynomial equations and inequalities. The timings and the number of
boxes returned for the two software
are presented in Table~\ref{tab:applications}.

\begin{table}
  \centering
  \caption{Subdivision solvers used to enclose the zero-set of the
  systems $(R)$ and $(A)$. For the system $(R)$, the subdivision process
was stopped when boxes were smaller than $2^{-4} \simeq 0.06$, both in
\texttt{ibexsolve} and \texttt{voxelize}.}
  \label{tab:applications}
  \begin{tabular}{|l|c|c|c|c|}
  \hline
  \multicolumn{1}{|c|}{\textbf{Software}} &
  \multicolumn{2}{c|}{\textbf{Robotics $(R)$}} &
  \multicolumn{2}{c|}{\textbf{Automatic $(A)$}} \\ \hline
  \textbf{} & \textbf{Time} & \textbf{Number of boxes} & \textbf{Time} & \textbf{Number of boxes} \\ \hline
  \texttt{ibexsolve} & 103s & 29871 & 2.5s & 0 \\ \hline
  \texttt{voxelize} & 0.1s & 7228 & 1.2s & 0 \\ \hline
\end{tabular}\\[1em]
\end{table}

We can see that both solvers could detect that the system $(A)$ has no
complex solutions of moduli less than $1$. In both cases,
\texttt{voxelize} was faster than \texttt{ibexsolve}, and significantly
faster for
the system $(R)$. This shows that the amortized evaluation scheme based
on the CSF data structure is efficient not only in theory, but also in
practice. On the other hand, \texttt{ibexsolve} and \texttt{voxelize} solve the system $(A)$ with a
time within the same order of magnitude, despite the fact that
\texttt{ibexsolve} does not used amortized evaluations. This might be
due to the fact that the contractors used by \texttt{ibexsolve} work well for this system. Remark
that it could be possible to combine contractors and amortized
evaluation scheme. The main issue is that after applying a contractor,
the boxes are not anymore aligned on a grid. This could be solved by
snapping the boxes to expanded boxes from a refined grid after applying the
contractors.

\begin{credits}
\subsubsection{\ackname}
The author wishes to thank Luc Jaulin, Thomas Cluzeau and Alban Quadrat
for their insightful remarks and examples discussed in this article.
\end{credits}

% REFERENCES
%\subsection*{References}
\bibliographystyle{splncs04}
\bibliography{bib-algcurves-utf8}

%\begin{description}
%
%  \item[{[1]}]
%{\sc J.~M. Snyder}, {\em Interval analysis for computer graphics}, in
%  Proceedings of the 19th annual conference on Computer graphics and
%  interactive techniques, SIGGRAPH '92, New York, NY, USA, 1992, ACM,
%  pp.~121--130.
%
%  \item[{[2]}]
%{\sc S.~Plantinga and G.~Vegter}, {\em Isotopic approximation of implicit
%  curves and surfaces}, in SGP '04: Eurographics/ACM SIGGRAPH Symposium on
%  Geometry Processing, 2004, pp.~245--254.
%
%\item[{[3]}]
%{\sc R.~B. Kearfott}, {\em Rigorous global search: continuous problems},
%  Nonconvex optimization and its applications, Kluwer Academic Publishers,
%  Dordrecht, Boston, 1996.
%
%\item[{[4]}]
%{\sc A.~Neumaier}, {\em Interval methods for systems of equations}, Cambridge
%  University Press, 1990.
%
%\item[{[5]}]
%{\sc L.~Jaulin, M.~Kieffer, O.~Didrit, and E.~Walter}, {\em {Applied Interval
%  Analysis with Examples in Parameter and State Estimation, Robust Control and
%  Robotics}}, {Springer London Ltd}, Aug. 2001.
%
%\item[{[6]}]
%{\sc S.~Smith and G.~Karypis}, {\em Tensor-matrix products with a compressed
%  sparse tensor}, in Proceedings of the 5th Workshop on Irregular Applications:
%  Architectures and Algorithms, IA3 '15, ACM, 2015, pp.~5:1--5:7.
%
%\item[{[7]}]
%{\sc S.~Chou, F.~Kjolstad, and S.~Amarasinghe}, {\em Format abstraction for
%  sparse tensor algebra compilers}, Proc. ACM Program. Lang., 2 (2018),
%  pp.~123:1--123:30.
%
%
%\end{description}

\end{document}